\documentclass[11pt]{article}
\usepackage{inputenc}
\usepackage{tikz}
\usetikzlibrary{arrows,calc,shapes,decorations.pathreplacing,positioning}
\usepackage{pgfplots}
\usepackage{comment}
\usepackage{amsmath}
\pgfplotsset{compat=newest}
\usepackage{amssymb}
\usepackage[hidelinks]{hyperref}
\usepackage{algorithm}
\usepackage{algpseudocode}

\newcommand{\commentout}[1]{}

\renewcommand{\eqref}[1]{Equation~(\ref{#1})}

\newcounter{remark}[section]

\def\claim{\par\medskip\noindent\refstepcounter{remark}\hbox{\bf Remark \arabic{section}.\arabic{remark}}
	\ 
}
\def\endclaim{
	\par\medskip}
\newenvironment{remark}{\claim}{\endclaim}

\newtheorem{theorem}{Theorem}

\newtheorem{corollary}{Corollary}

\newtheorem{definition}{Definition}
\def\endpf{{\ \hfill\hbox{\vrule width1.0ex height1.0ex}\parfillskip 0pt
	}}
	\newenvironment{proof}{\noindent{\bf Proof:}}{\endpf}

\begin{document}
		
		\title{Perturbation Robust Stable Matching}
		\author{Royi Jacobovic\footnote{Statistics Department, The Hebrew University of Jerusalem.}  \footnote{Comments may be sent to \textit{royi.jacobovic@mail.huji.ac.il}.}}
		\maketitle
		\begin{abstract}
			A well known result states that stability criterion for matchings in two-sided markets doesn't ensure uniqueness. This opens the door for a moral question with regard to the optimal stable matching from a social point of view. Here, a new notion of social optimality is proposed. Its novelty is its ability to take into account the possibility of agents to leave the matching after it has already been established. To formalize this real-life scenario, this work includes a well-defined probability model and a social cost function that maintain the general guidelines of leaving-agents situations. Finally, efficient algorithms to optimize this function are developed either under stability constraint or without it.     
			
		\end{abstract}
	
	\section{Introduction}
	Consider an instance of the classic stable-marriage problem (SMP), i.e. $n$ men and $n$ women who have strict nonnegative cost functions over the members of the opposite sex and themselves. Under these settings Gale and Shapley \cite{GS} proved that there exists at least one stable matching, i.e. a matching for which there is no blocking-pair of man and woman who prefer to be together rather than staying with their current partners. With respect to this result, Gale and Shapley provided an example which demonstrates that more than one stable matching can exist. This led several authors \cite{McVitie,Knuth,Y,Feder} to suggest several social cost functions in order to assess the quality of each stable matching and in particular to pinpoint the optimal ones from a social point of view.
	
	In this part of the work a new social cost function is suggested. In particular, it has been designed to take into account the criterion of robustness against leaving of agents due to exogenous motivations. As a motivating example, consider the case-study of the market for clinical-psychology graduate programs in Israel. In this market there are applicants who have preferences over the programs while each program ranks the applicants with respect to its own criteria. In practice, the matching-planner assigns each applicant to a specific program or to herself if she stays unmatched. Notice that a problem arises if only a short-time before the beginning of the academic year one applicant who has been matched to a graduate program is accepted by a top-leading program abroad. Allegedly, the matching-planner would be glad to match once again the applicants and programs that remained in the pool. The problem is that now both programs and applicants are obliged to their original matchings by commitments like jobs, rent contracts, thesis advisors, budget considerations and etc. Observe that these commitments are both expensive to be broken as well as hard to be re-established immediately in a new place, therefore high chance for changes in the original matching would decrease the market efficiency due to an increasing uncertainty. 
	
	With regard to this problem the rest is organized as follows: First of all, a probabilistic model which captures the possibility of agents to leave is introduced. Then, with respect to this model, the notion of perturbation-robust stable matching is defined. More specifically, considering the agents who stayed in the model, this notion is a stable-matching which controls the trade-off between minimizations of the expected social cost  and the expected distance to the new optimal matching. Section 2 includes an example which demonstrates the problem and the necessity to develop an efficient algorithm to solve it. Section 3 presents the details of a possible $\mathcal{O}[n^4\log(n)]$-time algorithm while section 4 describes how the same problem can be solved by a $\mathcal{O}(n^4)$-time algorithm where stability axiom is not required from the solution. Finally, section 5 is a conclusion which includes some optional directions for further research.    
	
	\subsection{Mathematical Model}
	
	For simplicity of notation and w.l.o.g consider the case when no more than a single agent may leave the model. Rigorously, denote the set of agents in this model by $\Omega$ and draw an agent $\omega_0$ from $\Omega\cup\phi$ w.r.t. some probability measure $p(\cdot)$ assigning a probability $p_\phi\geq0$ for the event that no one leaves the model and $p_{\omega_0}\geq0\ \ s.t. \ \ p_\phi+\sum_{\omega_0\in\Omega}p_{\omega_0}=1$ for the event that an agent $\omega_0\in\Omega$ is the only one to leave the model. Given this framework, for any $\omega_0\in\Omega$ who leaves the model define $\mu_{\omega_0}$ as the optimal matching that is returned by Feder's algorithm \cite{Feder} when its input is the stable marriage instance that doesn't include $\omega_0$. Similarly denote by $c_{\omega_0}(\mu)$ the vector of costs of all agents except $\omega_0$ under the matching $\mu$. Notice that when $\omega_0=\phi$ the notation $c_{\omega_0}(\mu)$ refers to the vector of all the costs of the agents while $\mu_{\omega_0}$ stands for the optimal matching associated with the original SMP instance. Now formally, with these notations in hands, make the following definition:
	
	\begin{definition}
		A matching $\mu^*$ is  a $\nu$-perturbation robust stable-matching iff it solves the following optimization over the set of stable matchings:
		
		\begin{equation}
		\mu^*\in\arg\min_\mu\Big\{\nu\mathbb{E}_p||c_{\omega_0}(\mu)||^2+(1-\nu)\mathbb{E}_p||c_{\omega_0}(\mu)-c_{\omega_0}(\mu_{\omega_0})||^2\Big\}
		\end{equation}
		
		where the norm is Euclidean, $\nu\in[0,1]$ is a known parameter and $\mu_{\omega_0}$ is the output of Feder's algorithm \cite{Feder} such that it minimizes the social cost function
		
		\begin{equation}
		\mu_{\omega_0}\in\arg\min_{\mu}||c_{\omega_0}(\mu)||^2
		\end{equation}
		
		over the set of stable matchings associated with the original problem instance minus $\omega_0$. 
		
	\end{definition}
	
	For clarity, denote the cost function which appears in equation (1) by $\psi(\cdot;\nu)$. Then, to give some intuition about its structure, notice that the first term stands for the social cost associated with $\mu$ while the second is a proxy for the social regret. Therefore, the interpretation of $\nu$ is of a parameter which reflects the preferences of the matching-planner over the trade-off between the social-cost and the risk of exceeding social demand for re-matching. Finally, notice that both terms share the same scale and determined with respect to the same norm, that is to make them comparable and let $\nu$ to be associated with clear interpretation.   
	
	\begin{remark}
		In Definition 1, $\mu_{\omega_0}$ has been specified as the output of the algorithm of Feder \cite{Feder}. While this specification fully discriminates $\mu_{\omega_0}$ from all other stable matchings that minimize the cost function mentioned in equation (2), it is not the most natural specification that can be considered. A more plausible definition for $\mu_{\omega_0}$ is made as follows: Consider the set of stable matchings which  minimize the cost function mentioned in equation (2). With regard to this set and given some matching $\mu$, let $\mu_{\omega_0}$ be the first member of this set (assume some alphabetic order over this set)  such that the squared norm $||c_{\omega_0}(\mu)-c_{\omega_0}(\mu_{\omega_0})||^2$ is minimized. Notice that while this definition seems to be making sense, it is not clear whether it defines $\mu_{\omega_0}$ that can be computed efficiently. Meanwhile until this computability question  is answered, observe that all of the next results can be easily modified to the case where  efficient solution exists for this problem.   
		
	\end{remark}
	
	\section{Motivating Example}
	This section presents an example which is based on the SMP instance which was introduced by Gale and Shapley \cite{GS}. Mainly, it demonstrates two points: First, that there exists a trade-off between minimizing the first and second terms of $\psi(\cdot;\nu)$. Secondly, that the optimization described in the previous section can't be solved systematically by existing algorithms. Namely, this example is a case where for one parametrization of $\nu$ the solution can be derived by execution of classic deferred-acceptance (DA) procedures \cite{GS,McVitie} but not by egalitarian social-cost optimizers \cite{ILG,Feder}. On the contrary, it is going to be shown that for another parametrization of $\nu$ the opposite situation holds.
	
	In details, consider the case where the costs of the agents are their rankings over the set of their possible partners, namely each agent attaches cost $k$ to the agent who is placed on the $k$'th place in her preference list. In addition, assume that $p_{m_1}=\frac{3}{4}=1-p_{\phi}$ while the preferences of the agents are given by:
	
	\begin{equation*}
	m_1: w_1\succ w_2\succ w_3\succ m_1
	\end{equation*}  
	\begin{equation*}
	m_2: w_2\succ w_3\succ w_1\succ m_2
	\end{equation*}
	\begin{equation*}
	m_3: w_3\succ w_1\succ w_2\succ m_3
	\end{equation*}
	\begin{equation*}
	w_1: m_2\succ m_3\succ m_1\succ w_1
	\end{equation*}
	\begin{equation*}
	w_2: m_3\succ m_1\succ m_2\succ w_2
	\end{equation*}
	\begin{equation*}
	w_3: m_1\succ m_2\succ m_3\succ w_3
	\end{equation*}
	
	Considering these preferences, as pointed by Gale and Shapley, there are three different stable matchings: The male optimal, female optimal and the egalitarian which are  given respectively by: 
	
	\begin{equation*}
	\mu_M=\{(m_1,w_1),(m_2,w_2),(m_3,w_3)\}
	\end{equation*}
	\begin{equation*}
	\mu_F=\{(m_1,w_3),(m_2,w_1),(m_3,w_2)\}
	\end{equation*}
	\begin{equation*}
	\mu_E=\{(m_1,w_2),(m_2,w_3),(m_3,w_1)\}
	\end{equation*}
	
	Moreover if $m_1$ leaves the model, it is possible to verify that both men-proposing and women-proposing algorithms which were described by Gale and Shapley in \cite{GS} return the same output. Hence, as a consequence, it can be deduced that there is one unique stable matching for the remaining agents which is given by 
	
	\begin{equation*}
	\mu_{m_1}=\{(w_1,w_1),(m_2,w_2),(m_3,w_3)\}
	\end{equation*}

	To start the examination of this example consider the parametrization $\nu=1$ which means that the matching-planner is totally focused on minimizing the expected sum of squares of the agents costs. Given this parametrization the following calculation shows that $\mu_E$ is the 1-perturbation robust stable matching:
	
	\begin{equation*}
	\psi(\mu_M;\nu=1)=\frac{1}{4}(1^2+1^2+1^2+3^2+3^2+3^2)+\frac{3}{4}(1^2+1^2+3^2+3^2+4^2)=34.5
	\end{equation*}
	\begin{equation*}
	\psi(\mu_E;\nu=1)=\frac{1}{4}(2^2+2^2+2^2+2^2+2^2)+\frac{3}{4}(2^2+2^2+2^2+2^2+4^2)=30
	\end{equation*}
	\begin{equation*}
	\psi(\mu_F;\nu=1)=\frac{1}{4}(3^2+3^2+3^2+1^2+1^2+1^2)+\frac{3}{4}(3^2+3^2+1^2+1^2+4^2)=34.5
	\end{equation*}

	Notice that if no one leaves the model , then simple calculations lead to the result that the optimal stable matching is $\mu_\phi=\mu_E$. If so, the following calculations show why $\mu_M\neq\mu_E$ is the $0$-perturbation robust stable matching: 
	
	\begin{equation*}
	\psi(\mu_M;\nu=0)=\frac{1}{4}(1^2+1^2+1^2+1^2+1^2+1^2)+\frac{3}{4}(0^2+0^2+0^2+0^2+1^2)=2\frac{1}{4}
	\end{equation*}
	\begin{equation*}
	\psi(\mu_E;\nu=1)=\frac{1}{4}(0^2+0^2+0^2+0^2+0^2+0^2)+\frac{3}{4}(1^2+1^2+1^2+1^2+2^2)=6
	\end{equation*}
	\begin{equation*}
	\psi(\mu_F;\nu=0)=\frac{1}{4}(1^2+1^2+1^2+1^2+1^2+1^2)+\frac{3}{4}(2^2+2^2+2^2+2^2+3^2)=20.25
	\end{equation*}
	
	Observe that $\mu_E$ is returned by the social cost minimizers \cite{Feder,ILG} while $\mu_M$ is the output of \cite{GS,McVitie}. Therefore this example shows that none of the classic DA procedures or the social cost minimizers can be applied directly to solve the $\nu$-perturbation robust stable matching problem systematically. Finally, it is of special importance to see how this example exemplifies the trade-off between the minimizations of the two terms of the function $\psi(\cdot;\nu)$.  
	
	\section{Algorithmic Solution}
	
	To enhance the curiosity about the problem depicted so far from an algorithmic point of view, observe the result made by Irving and Leather \cite{Irving}. Practically these authors pointed out an instance of the stable marriage problem which is associated with  an exponential nmber (in $n$) of different stable matchings. With regard to their interesting result, some social cost minimizations over the set of stable matching were found to be NP-hard \cite{Kato,Feder,Y} while for others, efficient algorithms were found \cite{GS,McVitie,Gusfield,ILG,Feder}. This section shows how the algorithmic approach of Irving, Leather and Gusfield \cite{ILG} can be modified in order to solve the $\nu$-perturbation stable matching problem in $\mathcal{O}[n^4log(n)]$-time.\footnote{ For reference, observe that the DA procedure and Feder's algorithm are associated with $\mathcal{O}(n^2)$ and $\mathcal{O}(n^3)$-time solutions respectively.} This section is consisted of three stages: To begin with, the term of rotations is reviewed together with some important and well-known results about its relation to the set of stable matchings. The second stage is focused on implementation of these results in order to have an equivalent graph representation of the optimization appearing in equation (1). Finally, it is shown that the same arguments of Irving \textit{et al} \cite{ILG} on the graph defined in stage 2 leads to an efficient solution.
	
	\subsection{Rotations}
	The following definitions and statements are sequentially made in order to recall the representation of the set of stable matchings in terms of rotations as done by Irving \textit{et al} \cite{Irving}.  
	
	\begin{definition}
		Considering the men-proposing Gale-Shapley algorithm \cite{GS}, define for each agent a shortlist in the following way:
		\begin{enumerate}
			\item Create a sorted list of the members of the opposite sex  with respect to the agent's preference list from the most preferred at the first place to the most undesired in the last. 
			\item If the agent is a male, delete from the list created in 1 all females who refused to the agent's proposal during the DA procedure.
			
			\item Otherwise, if the agent is a female, delete from the list created in 1 all males who were denied by the agent during the DA procedure. 
			\item Return the list to the user. 
		\end{enumerate}
	\end{definition}

	\begin{definition}
		A sequence of pairs $\rho:=\{(m_{i_0},w_{i_0}),\ldots,(m_{i_{r-1}},w_{i_{r-1}})\}$ is a rotation. Moreover, the rotation $\rho$ is exposed in the shortlists iff:
		\begin{enumerate}
			\item For any $k=0,\ldots,r-1$, $w_{i_k}$ is located in the first place of the shortlist of $m_{i_k}$.
			\item For any $k=0,\ldots,r-1$ (taken modulo $r$), $w_{i_{k+1}}$ is located in the second place in the shortlist of $m_{i_k}$.
		\end{enumerate}
		
	\end{definition}
	
	\begin{definition}
		A rotation $\rho$ which is exposed in the shortlists is said to be eliminated (or relaxed) iff:
		\begin{enumerate}
			\item For any $k=0,\ldots,r-1$ (taken modulo $r$), $m_{i_k}$ is matched with $w_{i_{k+1}}$. 
			\item For any $k=0,\ldots,r-1$, delete $w_{i_{k}}$ and $m_{i_{k}}$ from each other's shortlists. 
		\end{enumerate}
	\end{definition}
	
	Very intuitively, by making simple examples it is possible to see that there is some order in which the rotations can be exposed and eliminated sequentially. Formally, define this order as follows:
	
	\begin{definition}
		The rotational partial relation $\succ$ is defined as follows: For any two rotations $\pi$ and $\rho$, note $\pi\succ\rho$ iff $\rho$ must be eliminated before $\pi$ is exposed in the shortlists.  
	\end{definition}
	
	\begin{definition}
		The rotation poset (partially ordered set) is given by $P:=(R,\succ)$ where $R$ is the set of all possible rotations and  $\succ$ is the rotational partial relation. 
	\end{definition}
	
	\begin{definition}
		$C$ is a closed subset of $P$ iff $\rho\in C$ implies that any $\pi\prec\rho$ satisfies $\pi\in C$ as well.  
	\end{definition}
	
	At this stage, after establishment of this framework, recall the eminent result made by Irving and Leather \cite{Irving}. Notice that the strength of this result is by providing an option to express the set of stable matchings in terms of rotations and vice versa.
	
	\begin{theorem}[Irving and Leather-1987]
		There exists one to one correspondence between the set of all closed subsets of $P$ and the set of stable matchings.
	\end{theorem}
	
	\subsection{Rephrased Optimization}
	To re-formulate the optimization (1) on a graph to be later defined, consider a fixed $\nu\in[0,1]$ and an arbitrary rotation $\rho:=\{(m_{i_0},w_{i_0}), \ldots,(m_{i_{r-1}},w_{i_{r-1}})\}$ which is exposed in the shortlists of a matching $\mu$. Then, make the notation $\mu/\rho$ for the matching created from $\mu$ by the elimination of $\rho$. With respect to these notations, the change in the social cost function which is occurred by the elimination of $\rho$ is defined as follows
	
	\begin{equation}
	W(\rho;\mu):=\psi(\mu/\rho;\nu)-\psi(\mu;\nu)
	\end{equation}
	
	The next theorem claims that this change is invariant to the matching from which $\rho$ is eliminated from. 
	
	\begin{theorem}
		If $\mu_1$ and $\mu_2$ are two stable matchings in which the rotation $\rho$ is exposed in their shortlists, then $W(\rho;\mu_1)= W(\rho;\mu_2)$
	\end{theorem}
	
	\begin{proof}
		 Let $\mu$ be a stable matching in which the rotation $\rho$ is exposed in its shortlist and make the notation $c_{\omega_0}(\omega;\mu)$ for the cost of the agent $\omega$ in the cost vector $c_{\omega_0}(\mu)$. Using this notation, the guideline of this proof is to show that $W(\rho;\mu)$ is a function which is determined uniquely by the costs of the agents who belong to $\rho$. To this end, observe that for any agent $\omega\in\Omega$ who isn't paired in $\rho$,  $c_{\omega_0}(\omega;\mu)=c_{\omega_0}(\omega;\mu/\rho)$ and hence
		 
		 \begin{equation}
		 \end{equation}
		 \begin{equation*}
		 W(\rho;\mu_1)=\sum_{w_0\in\Omega\cup\phi}p_{\omega_0}\Big[\nu||c_{\omega_0}(\mu/\rho)||^2+(1-\nu)||c_{\omega_0}(\mu/\rho)-c_{\omega_0}(\mu_{\omega_0})||^2\Big]-
		 \end{equation*}
		 \begin{equation*}
		 -\sum_{w_0\in\Omega\cup\phi}p_{\omega_0}\Big[\nu||c_{\omega_0}(\mu)||^2+(1-\nu)||c_{\omega_0}(\mu)-c_{\omega_0}(\mu_{\omega_0})||^2\Big]=
		 \end{equation*}
		 \begin{equation*}
		 =\sum_{w_0\in\Omega\cup\phi}p_{\omega_0}\sum_{\omega\in\Omega/\{\omega_0\}}\Big\{\nu[c_{\omega_0}^2(\omega;\mu/\rho)-c_{\omega_0}^2(\omega;\mu)]+
		 \end{equation*}
		 \begin{equation*}
		 +(1-\nu)[c_{\omega_0}^2(\omega;\mu/\rho)-c_{\omega_0}^2(\omega;\mu)-2c_{\omega_0}(w;\mu_{\omega_0})\cdot(c_{\omega_0}(\omega;\mu/\rho)-c_{\omega_0}(\omega;\mu))]\Big\}=
		 \end{equation*}
		 \begin{equation*}
		 =\sum_{w_0\in\Omega\cup\phi}p_{\omega_0}\sum_{\omega\in\rho}\Big\{\nu[c_{\omega_0}^2(\omega;\mu/\rho)-c_{\omega_0}^2(\omega;\mu)]+
		 \end{equation*}
		 \begin{equation*}
		 +(1-\nu)[c_{\omega_0}^2(\omega;\mu/\rho)-c_{\omega_0}^2(\omega;\mu)-2c_{\omega_0}(w;\mu_{\omega_0})\cdot(c_{\omega_0}(\omega;\mu/\rho)-c_{\omega_0}(\omega;\mu))]\Big\}
		 \end{equation*}
		 
		 Finally, since the internal sum is only over $\omega$'s that are belonged to $\rho$, it becomes clear why the change is uniquely determined by $\rho$.  	 
	\end{proof}

	The following definition is of the rotation poset weighted directed acyclic graph (DAG) which is later going to be helpful in the reformulation of the $\nu$-perturbation robust stable matching optimization. 
	
	\begin{definition}
		The rotation poset weighted-DAG $G(P):=(V,E,W)$ is defined as follows:
		
		\begin{enumerate}
			\item The set of vertices is the set of rotations, i.e. $V=R$.
			\item A pair $(\rho,\pi)\in R\times R$ is associated with a directed edge from $\rho$ to $\pi$ in $E$ iff $\rho\prec\pi$ and there is no $\tau\in R$ such that $\rho\prec\tau\prec\pi$.
			\item For each vertex $\rho\in R$ there is a weight  $-W(\rho)$.   
		\end{enumerate} 
		
	\end{definition}
	
	\begin{definition}
		A subset of vertices of a DAG is called closed iff for each node in the set, all the parents are in the set.
	\end{definition}
	
	\begin{corollary}
		The minimization defined by equation (1) is equivalent to finding the maximal weight closed subset of $G(P)$.
	\end{corollary}
	
	\begin{proof}
		Notice that the set of closed subsets of $G(P)$ is exactly the same as the set of all closed subsets of $P$. Then, by Theorem 1, deduce that there exists one to one correspondence between the set of all closed subsets of $G(P)$ and the set of all stable matchings. Now, because the weight of each vertex is the minus of the change in the social cost occurred by 'passing' in this vertex (eliminating the vertex rotation), then the result is that the total weight of any closed subset of $G(P)$ is the difference in the values of $\psi(\cdot;\nu)$ which are associated respectively with the stable-matching defined by the closed subset and the male-optimal stable matching. If so, the result is that the maximal weight closed subset of $G(P)$ is the one which solves the minimization appeared in equation (1). 
		
	\end{proof}

	\subsection{$\mathcal{O}[n^4log(n)]$-Time Algorithm}
	
	The algorithm to be described is consisted of two stages: First, perform a computation of the set of weights $W$ by the formula described in equation (4). Second, given the set $W$, find the maximal-weight closed subset of $G(P)$ by the method described in \cite{ILG}. To analyze the complexity of this algorithm, it is shown that the first stage takes  $\mathcal{O}(n^4)$-time. Observe that for convenience, there is a review about the arguments of \cite{ILG} which demonstrate how the second stage can be done in $\mathcal{O}[n^4log(n)]$-time.  To start with, recall a known result made by Irving \textit{et al} \cite{ILG}:
	
	\begin{theorem}[Irving Leather and Gusfield-1987]
		\ \
		
		\begin{enumerate}
			\item $|R|=\mathcal{O}(n^2)$
			\item The set of all rotations $R$ can be computed in $\mathcal{O}(n^3)$-time.
		\end{enumerate}	
		
	\end{theorem}

	\begin{theorem}
		The set $\{\mu_{\omega_0}\}_{\omega_0\in\Omega}$ can be computed in $\mathcal{O}(n^4)$-time. 
	\end{theorem}
	
	\begin{proof}
		
		By Feder $\cite{Feder}$, for each $\omega_0\in\Omega
		\cup\phi$, $\mu_{\omega_0}$ can be found in $\mathcal{O}(n^3)$-time. Recalling that $|\omega_0\in\Omega
		\cup\phi|=\mathcal{O}(n)$ then $\mathcal{O}(n^4)$-time is needed in order to compute $\{\mu_\omega\}_{\omega\in\Omega\cup\phi}$. 
		
	\end{proof}

	\begin{theorem}
		The set $\{W(\rho);\rho\in R\}$ can be computed in $\mathcal{O}(n^4)$-time.
	\end{theorem}
	
	\begin{proof}
		First compute the set $\{\mu_\omega\}_{\omega\in\Omega\cup\phi}$ in $\mathcal{O}(n^4)$-time. Then, for each $\rho\in R$ compute $W(\rho)$ by equation (4). As one can easily see, this expression involves double summation whose calculation is bounded by $\mathcal{O}(n^2)$ arithmetic operations. Recalling that $|R|=\mathcal{O}(n^2)$ deduce that $\{W(\rho);\rho\in R\}$ can be computed in $\mathcal{O}(n^4)$.
		
	\end{proof}
	
	To find a solution for the $\nu$-perturbation-robust stable-matching optimization, one can use the same technique used by Irving \textit{et al} \cite{ILG}: Construct a sparse subgraph $G'(P)$ of $G(P)$ which preserves all closed subsets of $G(P)$. By doing so, it is enough to find the maximal weighted closed subset of $G'(P)$. Notice that since $G'(P)$ is sparse, then this task should be easier from a computational point of view. In details, $G'(P)$ shares the same set of vertices with $G(P)$. The set of edges of $G'(P)$ is formally defined in \cite{ILG}. Notice that for each edge which belongs to $G'(P)$, its weight is identical to the weight of the same edge in $G(P)$. For simplicity of notation denote the set of vertices, edges and weights of $G'(P)$ respectively by $V'$, $E'$ and $W'$. Then, the main properties of $G'(P)$ are summarized by the following theorem:
	
	\begin{theorem}[Irving, Leather and Gusfield-1987] 
		\ \
		\begin{enumerate}
			\item $|V'|,|E'|\leq\mathcal{O}(n^2)$
			\item The unweighted directed acyclic graph $G'(P)$ can be computed in $\mathcal{O}(n^3)$-time.
			\item The transitive closure of $G'(P)$ is $G(P)$.
		\end{enumerate}
		
	\end{theorem}
	
	\begin{corollary}
		The construction of $G'(P)$ can be done in $\mathcal{O}(n^4)$.
	\end{corollary}  
	\begin{proof}
		By statement 2 of Theorem 6, $(V',E')$ can be computed in $\mathcal{O}(n^3)$-time. In addition, since the weights of $G'(P)$ are the same as in $G(P)$ for any edge in $V'\subseteq V$, then by Theorem 5, $W'$ can be computed in $\mathcal{O}(n^4)$-time. Therefore, $G'(P)=(V',E',W')$ can be computed in $\mathcal{O}(n^4)$. 
	\end{proof}

	\begin{corollary}
		In order to solve the optimization phrased by equation (1), it's enough to find a maximal-weight closed subset of $G'(P)$. 
	\end{corollary}
	
	\begin{proof}
		As it has already been discussed, it is enough to find a maximal-weighted closed subset of $G(P)$ in order to solve the optimization phrased by equation (1). By statement 3 of Theorem 6,  $G(P)$ is the transitive closure of $G'(P)$ and hence recalling the reachability property of the transitive closure of a DAG, $G'(P)$ preserves all of $G(P)$ closed subsets.   Therefore, it is enough to find a maximal-weight closed subset of $G'(P)$. 
		
	\end{proof}
	
	Finally to show that the optimization given by equation (1) can be solved in $\mathcal{O}[n^4log(n)]$-time, recall another claim of Irving et al \cite{ILG} which states that finding the maximal-weight closed subset of $G'(P)$ is equivalent to a specific max-flow problem to be next specified. 
	
	\begin{definition}
		Define the capacitated $[s-t]$ graph $G'_{s,t}(P)$ by the following procedure;
		\begin{enumerate}
			\item connect a source $s$ and a sink $t$ to the graph $G'(P)$.
			\item For each vertex $\rho$ in $G'(P)$ with negative weight, i.e. $W(\rho)<0$ make a directed edge from $s$ to $\rho$ with weight $|W(\rho)|$
			\item For any vertex $\rho$ with positive weight, i.e. $W(\rho)>0$ make a directed edge from  $\rho$ to $t$ with weight $W(\rho)$.
			\item Set the weights of all edges of the original $G'(P)$ to be infinity.
		\end{enumerate}      
		
	\end{definition}
	
	\begin{theorem}[Irving, Leather and Gusfield-1987]
		The number of edges $\mathcal{E}$ and the number of nodes $N$ of $G'_{s,t}(P)$ are both bounded by $\mathcal{O}(n^2)$. 
	\end{theorem}
	\begin{theorem}[Irving, Leather and Gusfield-1987]
		Let $X$ be the set of edges crossing a minimum $s-t$ cut in $G'_{s,t}(P)$, and denote the weight of $X$ by $W(X)$. The positive nodes in the maximum-weight closed subset of $G'(P)$ are exactly the positive nodes whose edges into $t$ are uncut by $X$. These nodes and all the nodes that reach them in $G'(P)$, define a maximum-weight closed subset in $G'(P)$. 
	\end{theorem}
	\begin{theorem}
		The optimization which defines the $\nu$-perturbation robust stable matching, can be solved in $\mathcal{O}[n^4log(n)]$-time.
	\end{theorem}
	
	\begin{proof}
		By Theorem 8, given the graph $G'_{s,t}(P)$, the optimization phrased in equation (1) can be reduced to min-cut max-flow problem defined over  $G'_{s,t}(P)$. By Theorem 7,  $G'_{s,t}(P)$ include no more than $\mathcal{O}(n^2)$ nodes and no more than $\mathcal{O}(n^2)$ edges. Therefore, the max-flow algorithm of Sleator and Tarjan \cite{tarjan} implies that it can be solved in time of $\mathcal{O}[N\mathcal{E}log(N)]$, i.e. its solution can be obtained in $\mathcal{O}[n^4log(n)]$-time. Since Corollary 3 states that $G'(P)$ can be constructed in $\mathcal{O}(n^4)$-time, it is easy to verify that $G'_{s,t}(P)$ can be constructed in the same order of time. Therefore, the concluding result is that the optimization given in equation (1) can be solved in $\mathcal{O}[n^4log(n)]$.

	\end{proof}
	
	\section{Relaxed Perturbation Robust Matching}
	A known and very intuitive observation which was made by Blum et al \cite{Blum} states that given a stable matching of the original SMP instance, the leaving of an agent $\omega_0\in\Omega$ might lead to instability of the new matching. This puts a big question mark upon the stability requirement in the definition of the $\nu$-perturbation robust stable matching. To justify the original definition,  two possible arguments can be made: First, making the original matching stable lets the user to control the number of blocking-pairs in the new matching created after the leaving of $\omega_0$. Second, the user might be willing to ensure that if no one leaves the model, then it will be stable. 
	
	While both of these arguments are reasonable, still they are not definite, i.e. it is sensible to look for the matching which solves the optimization defined by equation (1) over the set of all valid matchings (and not necessary stable). To distinguish this matching from the former which has been discussed so far, call it the \textit{relaxed} $\nu$-perturbation robust matching. In this section it is  shown how the problem of finding it can be transferred into an equivalent assignment problem and hence to be associated with efficient solution. As a first step, recall the formal description of the classical assignment problem \cite{assignment} in the context of the probabilistic model discussed so far.
	
	\begin{definition}[Assignment Problem]
		Given the set of agents $\Omega$  with a cost function $f:\Omega\times \Omega\rightarrow\mathbb{R}_+$, the assignment problem is the maximization which is given by
		
		\begin{equation*}
		max:\sum_{(\omega_1,\omega_2)\in\Omega\times\Omega}-f(\omega_1,\omega_2)\cdot z_{\omega_1,\omega_2}
		\end{equation*}
		\begin{equation*}
		s.t: \sum_{\omega_2\in\Omega} z_{\omega_1,\omega_2}=1\ \ ,\forall \omega_2\in\Omega
		\end{equation*}
		\begin{equation*}
		\ \  \ \  \ \ \sum_{\omega_1\in\Omega}z_{\omega_1,\omega_2}=1\ \ ,\forall \omega_1\in\Omega
		\end{equation*}
		\begin{equation*}
		\ \ \ \ \ \ \ \ \ \  z_{\omega_1,\omega_2}\in\{0,1\}\ \ ,\forall \omega_1,\omega_2\in\Omega
		\end{equation*}
		
	\end{definition}  
	
	\begin{remark}
		Notice that $f(\cdot)$ is a general cost function that, given $\Omega$, determines the assignment problem. Don't confuse it with the cost functions $c(\omega,\mu),\omega\in\Omega$ that return the cost suffered by the agents when the matching $\mu$ is implemented.
	\end{remark}
	Now, the next corollary defines the function $f(\cdot)$ in an appropriate way so the respective assignment problem  becomes equivalent to the generalized $\nu$-perturbation robust optimization. For simplicity of notation, for any $(\omega_1,\omega_2)\in\Omega^2$, denote the cost of agent $\omega_1$ from being matched to $\omega_2$ by $c(\omega_1,\omega_2)$. In addition, observe that for any matching $\mu$ and an agent $\omega\in\Omega$, $\mu(\omega)$ is the agent that is matched to $\omega$ by the matching $\mu$. 
	
	\begin{corollary}
		Note the sex of any agent $\omega\in\Omega$ by $S(\omega)$ and define the function $f(\cdot)$ as follows:
		
		\begin{enumerate}
			\item  If $S(\omega_1)=S(\omega_2)\wedge\omega_1\neq\omega_2$, then $f(\omega_1,\omega_2):=\infty$
			
			\item If $S(\omega_1)\neq S(\omega_2)$, then
			
			\begin{equation*}
			f(\omega_1,\omega_2):=
			\end{equation*}
			\begin{equation*}
			=\nu\Big\{(1-p_{\omega_1}-p_{\omega_2})\big[c^2(\omega_1,\omega_2)+c^2(\omega_2,\omega_1)\big]+p_{\omega_1}\cdot c^2(\omega_2,\omega_2)+p_{\omega_2}\cdot c^2(\omega_1,\omega_1)\Big\}+
			\end{equation*}
			\begin{equation*}
			+(1-\nu)\Big\{(1-p_{\omega_1}-p_{\omega_2})\mathbb{E}_p\big[(c(\omega_1,\omega_2)-c(\omega_1,\mu_{\omega_0}(\omega_1)))^2+(c(\omega_2,\omega_1)-
			\end{equation*}
			\begin{equation*}
			-c(\omega_2,\mu_{\omega_0}(\omega_2)))^2|\omega_0\neq \omega_1,\omega_2\big]+p_{\omega_1}\big[c(\omega_2,\omega_1)-c(\omega_2,\mu_{\omega_1}(\omega_2))\big]^2+
			\end{equation*}
			\begin{equation*}
			+p_{\omega_2}\big[c(\omega_1,\omega_2)-c(\omega_1,\mu_{\omega_2}(\omega_1))\big]^2\Big\}
			\end{equation*}
			
			\item Otherwise,
			\begin{equation*}
			f(\omega,\omega):=(1-p_\omega)\Big[\nu c^2(\omega,\omega)+(1-\nu)\mathbb{E}_p\big[\big(c(\omega,\omega)-c(\omega,\mu_{\omega_0}(\omega))\big)^2\big]\Big]
			\end{equation*} 
			
		\end{enumerate}
		
		If so, the respective assignment problem is equivalent to solving the relaxed $\nu$-perturbation robust matching problem.
	\end{corollary} 
	\begin{proof}
		To begin with, notice that the first clause in the definition of $f(\cdot)$ guarantees that no two agents from the same sex are paired by the algorithm. To explain this, pair each agent to herself and have finite value for the goal function. Then, with respect to the expression of $f(\cdot)$ and the definition of $\psi(\cdot;\nu)$, it is possible to use simple conditional expectation rules and the definition to get the needed result.
		
	\end{proof}
	
	\begin{theorem}[Edmonds, Karp and Hopcroft]
		The assignment problem can be solved in $\mathcal{O}(n^3)$-time.
	\end{theorem}
	
	\begin{theorem}
		The relaxed $\nu$-perturbation robust matching can be found in $\mathcal{O}(n^4)$-time
	\end{theorem}
	
	\begin{proof}
		By Theorem 5, the set $\{\mu_\phi\}\cup\{\mu_{\omega_0};\omega_0\ \ s.t. \ \ p_{\omega_0}>0\}$ can be computed in $\mathcal{O}(n^4)$-time. 
		Notice that given this set of matchings for any $(\omega_1,\omega_2)\in\Omega\times\Omega$, $f(\omega_1,\omega_2)$ equals to either $\infty$ or to another expression which involves conditional expectation which requires no more than $\mathcal{O}(n)$ arithmetic operations. In addition, since $|\Omega\times\Omega|=\mathcal{O}(n^2)$, then the result is that $f(\cdot)$ can be calculated in $\mathcal{O}(n^4)$-time. Finally, by Theorem 10, the assignment problem which is associated with $f(\cdot)$ can be solved in $\mathcal{O}(n^3)$-time and hence the needed result is done.
		
	\end{proof}

	\section{Conclusion}
	In this work a new notion of optimality for matchings is phrased. This new concept is based on the natural aspiration that given prior knowledge about the possibility of agents to leave, the matching maker would be willing to make some benefit from this knowledge. Here there are some cautions and points that can be extended by further research, regarding our research. 
	
	First of all, one should be aware to the fact that the running-times were calculated with respect to the simplification which states that no more than one agent can leave the model. In real world situations, less restrictive models can be preferred and hence the running-times should be greater respectively. Notice that the current approach leads to polynomial algorithm only if the number of subsets (not necessarily foreigns) of agents that can quit the model all together is polynomial as well. Therefore a well motivated challenge which stands in front of the statistical community is to look for a practical way to select and estimate such a model in a way that allow fair calibration.  
	
	Another aspect of the current results is providing more strength to the work of Irving \textit{et al} \cite{ILG}. To see this, notice that the current work presents a new well-motivated example that can be solved in the spirit of their algorithmic approach. Keeping with this guideline, it can be interesting to replace the Euclidean norm which appears in the definition of the $\nu$-perturbation robust stable matching by norms like $L_1$ or $L_\infty$ that are not smooth. Therefore an open question is to check whether these norms define a $\nu$-perturbation robust stable matching that can be solved efficiently. 
	
	Another direction for further research is to evaluate the price in terms of the social cost of the stability requirement in the original definition of the $\nu$-perturbation robust stable matching. Regarding this price, it seems reasonable to discuss the prices of anarchy and stability \cite{Nissan} with respect to the model specifications.
	
	To continue with, the market-design community is encouraged to generalize this model in the direction of endogenous leavings of agents, i.e agents who leave the model with probabilities that are non decreasing in their costs. Finally, notice that as mentioned by \cite{Mc}, the results of Roth et al \cite{Roth} enable to generalize the results of the current work to the case of SMP instances that include different numbers of men and women.
	\\
	\\
	\textbf{Acknowledgment:} I'd like to thank Assaf Romm, Or Zuk and Elisheva Schwarz for their comments.

	\end{document}